\documentclass[11pt,english]{article}
\usepackage{amsmath,amssymb,amsfonts}
\usepackage{paralist}
\usepackage{amsthm}
\usepackage{fullpage, comment}
 \usepackage{mathtools}
\usepackage{color}
 \usepackage{algorithm}
\usepackage[noend]{algorithmic}
\usepackage{enumitem}
\newtheorem{theorem}{Theorem}[section]
\newtheorem{lemma}[theorem]{Lemma}
\newtheorem*{lemma*}{Lemma}
\newtheorem{proposition}[theorem]{Proposition}
\newtheorem{lem}[theorem]{Lemma}
\newtheorem{fact}[theorem]{Fact}
\newtheorem{corollary}[theorem]{Corollary}

\newtheorem{claim}[theorem]{Claim}

\newtheorem{prop}[theorem]{Proposition}
\newtheorem{definition}[theorem]{Definition}

\theoremstyle{remark}
\newtheorem{remark}[theorem]{Remark}

\newcommand{\F}{\ensuremath{\mathbb{F}}}

\newcommand{\R}{\ensuremath{\mathbb{R}}}

\newcommand{\Z}{\ensuremath{\mathbb{Z}}}

\newcommand{\inner}[1]{\langle{#1}\rangle}

\newcommand{\mmod}{~\mathrm{mod}~}
\newcommand{\supp}[1]{\text{Support}(#1)}

\newif\ifnotes\notestrue
\ifnotes

\newcommand{\enote}[1]{\textcolor{red}{{\bf (Elena:} {#1}{\bf ) }} \marginpar{\tiny\bf
             \begin{minipage}[t]{0.5in}
               \raggedright ELENA NOTE
            \end{minipage}}}        
\newcommand{\gnote}[1]{\textcolor{blue}{{\bf (GV:} {#1}{\bf ) }} \marginpar{\tiny\bf
             \begin{minipage}[t]{0.5in}
               \raggedright ~
                \end{minipage}}}  
\newcommand{\knote}[1]{\textcolor{red}{{\bf (Karthik:} {#1}{\bf ) }} \marginpar{\tiny\bf
             \begin{minipage}[t]{0.5in}
               \raggedright ??
            \end{minipage}}}
\else
\newcommand{\gnote}[1]{}
\newcommand{\knote}[1]{}
\newcommand{\enote}[1]{}
\fi

\newcommand{\ignore}[1]{}
\newcommand{\sidenote}[1]{ \marginpar{\tiny\bf
             \begin{minipage}[t]{0.5in}
               \raggedright *
            \end{minipage}}}

\usepackage{hyperref}

\title{Deciding Orthogonality in Construction-A Lattices}
\author{Karthekeyan Chandrasekaran \thanks{Email: {\tt karthe@illinois.edu}}
\and Venkata Gandikota \thanks{Email: {\tt vgandiko@purdue.edu}}
\and Elena Grigorescu \thanks{Email: {\tt elena-g@purdue.edu}} 
}

\date{}

\begin{document}
\maketitle
%conf{1} means with appendix (icalp)
%conf{0} means with proofs in main body
\def\conf{1} 

\thispagestyle{empty}
\begin{abstract}

Lattices are discrete mathematical objects with widespread applications to integer programs as well as modern cryptography. A fundamental problem in both domains is the Closest Vector Problem (popularly known as CVP). It is well-known that CVP can be easily solved in lattices that have an orthogonal basis \emph{if} the orthogonal basis is specified. This motivates the orthogonality decision problem: verify whether a given lattice has an orthogonal basis. Surprisingly, the orthogonality decision problem is not known to be either NP-complete or in P. 

In this paper, we focus on the orthogonality decision problem for a well-known family of lattices, namely Construction-A lattices. 
These are  lattices of the form $C+q\Z^n$, where $C$ is an error-correcting $q$-ary code, and are studied in communication settings.
We provide a complete characterization of lattices obtained from binary and ternary codes using Construction-A that have an orthogonal basis. We use this characterization 
%leads 
to give an efficient algorithm to solve the orthogonality decision problem. Our algorithm also finds an orthogonal basis if one exists for this family of lattices. We believe that these results could provide a better understanding of the complexity of the orthogonality decision problem for general lattices. 
%, and further, of the closely related Lattice Isomorphism Problem. 

\iffalse
Lattices are discrete mathematical objects that have recently become crucial to modern cryptography.
Many existing cryptographic schemes are based on the assumed hardness of lattice problems, and they essentially exploit variants of the problem of finding lattice vectors close to a given target vector, known as the Closest Vector Problem. The Closest Vector Problem can be easily solved in lattices that have an orthogonal basis, if the orthogonal basis is specified. 

 Motivated by the fact that  the fundamental  problem of deciding whether a given lattice has an orthogonal basis is not known to be in P or NP-complete, in this paper we focus on a well-known family of lattices, namely Construction-A lattices. 
These are  lattices of the form $C+q\Z^n$, where $C$ is an error-correcting code, and are well-motivated especially in communication settings. We characterize which Construction-A lattices are orthogonal,  when $q=2,3$. Based on this structural decomposition we show how to efficiently find an orthogonal basis if one exists. We believe that these families  could lead to a better understanding of the general question of deciding orthogonality, and further, of the related Lattice Isomorphism Problem.
\fi
\end{abstract}

%\newpage

\section{Introduction}
A lattice is the set of integer linear combinations of a set of basis vectors $B\in \R^{m\times n}$, namely $L = L(B) = \{ xB \mid x \in \mathbb{Z}^m \}$. 
Lattices are well-studied fundamental mathematical objects that have been used to model diverse discrete structures such as in the area of integer programming \cite{Kannan83}, or in factoring integers \cite{Schnorr13} and factoring rational polynomials \cite{LLL}.
In a groundbreaking result, Ajtai \cite{Ajtai96} demonstrated the potential of computational problems on lattices to cryptography,  by showing  average case/worst case equivalence  between lattice problems related to finding short vectors in a  lattice. 
This led to renewed interest in the complexity of two fundamental lattice problems: the Shortest Vector Problem (SVP) and the Closest Vector Problem (CVP). %They have been addressed by both the cryptography and the complexity theory community with their hardness in various approximation settings being used as a basis for security guarantees in cryptographic protocols.
%Ever since, the complexity of the Shortest Vector Problem (SVP) and of its related variant, the Closest Vector Problem (CVP), have become most studied in cryptography and complexity theory, and their hardness in various approximation settings have been used as a basis for security guarantees in cryptographic protocols. 
Concretely, in SVP, given a basis $B$ one is asked to output a shortest non-zero vector in the lattice, and in CVP, given a basis $B$ and a target $t\in \R^n$, one is asked to output a lattice vector closest to $t$.

Both SVP and CVP are NP-hard even to approximate up to subpolynomial factors (see \cite{MicReg-surv09} for a survey), and a great deal of research has been devoted to finding families of lattices for which SVP/CVP are easy. A simplest lattice for which CVP is easy is $\Z^n$: indeed, finding the closest lattice vector to a target $t\in \R^n$ amounts to rounding the entries of $t$ to the nearest integer. 
Surprisingly, 
%What is surprising is that  
given an arbitrary basis $B$, it is not known how to efficiently verify whether the lattice generated by $B$ is  isomorphic to $\Z^n$ upto an orthogonal transformation.  Further, given an arbitrary basis for a lattice, it is not known how to decide efficiently if the lattice has an orthogonal basis (an orthogonal basis is a basis in which all vectors are pairwise orthogonal). Similar to the case of $\Z^n$, having access to an orthogonal basis leads to an efficient algorithm to solve CVP, but finding an orthogonal basis given an arbitrary basis appears to be non-trivial, with no known efficient algorithms. 
 
 Deciding if a lattice is equivalent to $\Z^n$, and deciding if a lattice has an orthogonal basis, are special cases of the more general Lattice Isomorphism Problem (LIP). In LIP, given lattices $L_1$ and $L_2$ presented by their bases, one is asked to decide if they are isomorphic, meaning if there exists an orthogonal transformation that takes one to the other. LIP has been studied in \cite{PleskenS97, SikiricSV09, HavivR14} and is known to have a $n^{O(n)}$ algorithm \cite{HavivR14}. Recent results of \cite{LenstraS14, LenstraS14M} show that in certain highly symmetric lattices, isomorphism to $\Z^n$ can be decided efficiently.  
 % LIP is also related to the Code Equivalence Problems \cite{PetrankRoth}, where given two codes $C_1$ and $C_2$ over some finite field $F$ one is asked to determine if they are isomorphic under permutation of their coordinates. This problem was shown to be at least as hard as Graph Isomorphism \cite{PetrankRonth}, and can be solvable in $(2+o(1))^n$ \cite{Babai}.
The complexity of LIP is not well understood, and is part of the broader study of isomorphism between mathematical objects, of which Graph Isomorphism (GI) is a well-known elusive problem \cite{Babai-book}. Interestingly, there is a polynomial time reduction from GI to LIP \cite{SikiricSV09}.

Given that LIP, deciding isomorphism to $\Z^n$, and  deciding whether a lattice has an orthogonal basis appear to be difficult problems for general lattices, it is natural to address families of lattices where these problems are solvable efficiently. In this work, we focus on the problem of deciding orthogonality for a particular family of lattices, commonly known as Construction-A lattices \cite{ConwayS98}. Construction-A lattices $L$ are obtained from a linear error-correcting code $C$ over a finite field of $q$ elements (denoted $\F_q$) as $L=C+q\Z^n$. \footnote{The term `Construction-A' strictly refers to the case $q=2$,  but we will not make the distinction in this paper.} 
 We resolve the problem of deciding orthogonality in Construction-A lattices for $q=2$ and $q=3$ by showing an efficient algorithm. In addition, the algorithm outputs an orthogonal basis of the input lattice if such a basis exists. 
 
Our main technical contribution is a decomposition theorem for Construction-A lattices that admit an orthogonal basis. A natural way to obtain orthogonal lattices through Construction-A is by taking direct products of lower dimensional orthogonal lattices. We show that this is the only possible way.
We believe that our contributions are a step towards gaining a better understanding of lattice isomorphism problems for more general classes of lattices.

Extending our results to values $q>3$ might require new techniques. For larger $q$, a decomposition characterization seems to require a complete characterization of \emph{weighing matrices} of weight $q$ which is a known open problem. In particular, a direct product decomposition characterization of weighing matrices for the case of $q=4$ is known. However, even in this case the parts in the direct product decomposition may not be of constant dimension, so designing an efficient algorithm for the orthogonality decision problem through a direct product decomposition characterization appears to be non-trivial.

\subsection{Our results and techniques}

%In what follows we state our results. 
As mentioned above, we start by showing a structural decomposition of orthogonal lattices of the form $C+2\Z^n$ and $C+3\Z^n$ into constant-size orthogonal lattices. We remark that the decomposition holds up to permutations of the coordinates, and we use the notation $C_1\cong C_2$ and $L_1\cong L_2$ to denote the equivalence of codes and lattices under permutation of coordinates. We use the notation $L_1 \otimes L_2$ to denote the direct product of two lattices.
% Our decomposition result for orthogonal lattices of the form $C+2\Z^n$ is stated below.

\begin{theorem} \label{thm:binary-decomposition}
Let $L_C = C + 2\mathbb{Z}^n$ be a lattice obtained from a binary linear code $C \subseteq \mathbb{F}_2^n$. 
Then the following statements are equivalent: 
\begin{enumerate}
\item $L_C$ is orthogonal.
\item $L_C \cong \otimes_i L_i$, where each $L_i$ is 
either $\Z$,  
or $2\Z$,  
or the $2$-dimensional lattice generated by the rows of the matrix $\begin{bmatrix} 1&1\\1&-1\end{bmatrix}$.
%$$\begin{bmatrix} 1 & 1 \\ 1 & -1\end{bmatrix}.$$
\item $C \cong \otimes_i C_i$, where each $C_i $ is 
either a  length-$1$ binary linear code $\subseteq \{0,1\}$,
% the length-$1$ binary linear code $\{0,1\}$ 
%or the length-one binary linear code $\{0\}$
or the length-$2$ binary linear code $\{00,11\}$.
\end{enumerate}
\end{theorem}

The decomposition characterization leads to an efficient algorithm to verify if a given lattice obtained from a binary linear code using Construction-A is orthogonal. For the purposes of this algorithmic problem, the input consists of a basis to the lattice. The algorithm  finds the component codes given by the characterization thereby computing the orthogonal basis for such a lattice. %\enote{This is not an iff statement, so we probably still need to verify that the decomposition indeed is orthogonal over Z..., right?}

\begin{theorem}\label{thm:algo-binary}
Given a basis for a lattice $L$ obtained from a binary linear code 
%$C\subseteq \F_2^n$ 
using Construction-A, there exists an algorithm running in time $O(n^6)$ that verifies if $L$ is orthogonal, and if so, outputs an orthogonal basis. 
\end{theorem}

We obtain a similar decomposition and algorithm for lattices obtained from ternary codes. For succinctness of presentation we define the following integer matrix:
\begin{align*}
%M_2&=\begin{bmatrix} 1&1 \\ 1&-1\end{bmatrix}\\
M&=\begin{bmatrix} 1&1&1&0 \\ 1&-1&0&1\\ 1 & 0& -1&-1 \\ 0 & 1 & -1& 1\end{bmatrix}.
\end{align*}

%We have the following decomposition result for orthogonal lattices of the form $C+3\Z^n$:
\begin{theorem} \label{thm:ternary-decomposition}
Let $L_C = C + 3\mathbb{Z}^n$ be a lattice obtained from a ternary linear code $C \subseteq \mathbb{F}_3^n$. 
Then the following statements are equivalent: 
\begin{enumerate}
\item $L_C$ is orthogonal.
\item $L_C \cong \otimes_i L_i$, where each $L_i$ is either  $\Z$, or   $3\Z$, or the $4$-dimensional lattice generated by the rows of a matrix ${\cal{T}}(M)$ obtained from $M$ by negating some subset of columns.
\item $C \cong \otimes_i C_i$, where each $C_i $ is either a linear length-$1$ ternary code, 
%either the length-one ternary linear code $\{0,1,2\}$ 
%or the length-one ternary linear code $\{0\}$
or the linear length-$4$ ternary code generated by the rows of $({\cal{T}}(M) \mmod 3)\in \F_3^{4\times 4}$, where ${\cal T}(M)$ is  obtained from $M$ by negating some subset of its columns. 
\end{enumerate}
\end{theorem}

\begin{theorem}\label{thm:algo-ternary}
Given a basis for a lattice $L$ obtained from a ternary linear code using Construction-A, there exists an algorithm running in time $O(n^{8})$ that verifies if $L$ is orthogonal, and if so, outputs an orthogonal basis. 
\end{theorem}

Theorems ~\ref{thm:binary-decomposition} and \ref{thm:algo-binary} are proved in Section~\ref{sec: binary}. Theorems ~\ref{thm:ternary-decomposition} and \ref{thm:algo-ternary} are proved in Section~\ref{sec:ternary}.

% GV commented this out on 3Nov: 
%We defer all the missing proofs to the appendix.

%We have the analogous results for orthogonal lattices of the form $C+2\Z^n$ where $C$ is a binary linear code. 

\section{Preliminaries}
We denote the set of positive integers up to $n$ by $[n]$, the $n\times n$ identity matrix by $I_n$ and its $j^{th}$ row by $e_j$. For a vector $b\in \R^n$, let $b_j$ denote its $j^{th}$ coordinate, and $\|b\|$ be its $\ell_2$ norm.

%\subsection{Lattices}
%A $m$-dimensional lattice $L$ is a discrete additive subgroup of $\mathbb{R}^m$.
A lattice $L\subseteq \R^n$ is said to be of full rank if it is generated by $n$ linearly independent vectors. A lattice $L$ is said to be orthogonal if it has a basis $B$ such that the rows of $B$ are pairwise  orthogonal vectors. A lattice $L$ is {\em  integral} if it is contained in $\mathbb{Z}^n$, namely any basis for $L$ only consists of integral vectors. 
%A basis $B$ for a lattice $L$ is a set of linearly independent vectors such that every lattice vector can be expressed as an integer linear combination of the basis vectors, i.e., $L=\{ xB \mid x \in \mathbb{Z}^m \}$ . 
%A given lattice $L$ may have multiple bases.

We will denote by $\F_q$ a finite field with $q$ elements.
A {\em linear} code $C$ of length $n$ over  ${\mathbb F}_q$ is a vectorspace  $C\subseteq {\mathbb F}_q^n$. A linear code is specified by a generator matrix $G$ that consists of linearly independent vectors in $\mathbb{F}_q^n$. If $C\subseteq \F_2^n$, then it is called a {\em binary} code, and if $C\subseteq \F_3^n$, then it is called a {\em ternary} code.

The Construction-A of a lattice $L_C$ from a linear code $C\subseteq \F_q^n$, where $q$ is a prime, is defined as 
$L_C := \{ c + q \cdot z \mid c \in \phi(C), z \in \mathbb{Z}^n \}$, where   $\phi$ is the (real embedding) mapping $ i\in\F_q \mapsto i\in \Z$. Construction-A is often abbreviated as $L_C=C+q\Z^n$. 

For any vector $v = (v_1, \cdots, v_n) \in \Z^n$  define $v \mmod q :=(v_1 \mmod q , \cdots, v_n \mmod q)\in \F_q^n$.

\begin{claim}
\label{claim:c-is-Lmodq}
Let $q$ be a prime and L be an integral lattice. If $q\Z^n\subseteq L$ then $C=L \mmod q$ is a linear code over $\F_q$. 
%It can be shown that an integral lattice $L$ that contains $3\Z^n$ (or $2\Z^n$) gives a ternary (or binary) linear code $L\mod 3$ (or $L\mod 2$).
\end{claim}
\begin{proof}
Let $v \in L$ and   $v =( v \mmod q) + q z$ for some $z \in \Z^n$, where here we abuse notation and view $v \mmod q$ as embedded into the integers, instead of a vector in $\F_q^n$. Since $q \Z^n \subseteq L$, it follows that  $v-qz=v \mmod q \in L$.  To show that $C=L \mmod q$ is a linear code over $\F_q$, let $c_1, c_2\in C$. Then $c_1+c_2 \in L$ (where the addition is over $\Z$), and so $(c_1+c_2) \mmod q \in C$.
%we show that it is a  subspace of $\Z^n \mmod q$. Therefore, the embedding of $L \mmod q$ into $\F_q^n$ using $\phi^{-1}$ is a code. Let $c_1$ and $c_2$ be two vectors in $L \mmod q$. Since $c_1, c_2$ are also lattice vectors, and $L$ is closed under addition, $(c_1 + c_2) \in L$. Therefore, from the above argument, $(c_1 + c_2) \mmod q \in L \mmod q$. Also, for any $k \in \Z \mmod q$, $k c_1 \in L$ and therefore, $k c_1 \mmod q \in L \mmod q$. 
\end{proof}

 % Two lattices $L_1$ and $L_2$ are {\em equivalent}
 % if there exists a permutation matrix $P$ such that $L_1=L_2P$. In other words, lattices $L_1$ and $L_2$ are equivalent 
 %if they are identical up to a permutation of the coordinates. Similarly two codes $C_1$ and $C_2$ that are identical upto a permutation of the coordinates by $C_1\cong C_2$. 

%Let $\phi()$ be an embedding of $\mathbb{F}_q^n$ into $\mathbb{Z}^n$. \enote{is this for any embedding? natural is not a good word here} In this paper we consider  $\phi$ to be the mapping $ i\in\F_q \mapsto i\in \Z$. \gnote{ construction is defined only for this embedding $i \rightarrow i$ }

We will use the following immediate claim about product of lattices generated from codes. 
\begin{claim}\label{lemma:product_codes_lattices}
Let $L = C + q \mathbb{Z}^n$, for some $q$-ary linear code $C \subseteq \mathbb{F}_q^n$. 
If $L \cong L_1 \otimes L_2$, and $L_1\subseteq \Z^k$, then $L_1 \cong C_1 + q \mathbb{Z}^k$ and $L_2 \cong C_2 + q \mathbb{Z}^{n-k}$, for $q$-ary linear codes $C_1$ and $C_2$ that are projections of $C$ on the coordinates corresponding to $L_1$ and $L_2$ respectively. 
\end{claim}

A matrix $U$ is {\em unimodular}  if $U\in \Z^{n \times n}$ and $\det(U)\in \{\pm1\}$.
Two different bases $B_1, B_2$ give the same lattice  if and only if there exists a unimodular matrix $U$ % \in GL(n, \mathbb{Z})$
such that $B_1 = U B_2$. 
%A lattice $L$ generated by a set of vectors is the set of integer linear combination of the vectors. 
The {\em Hermite Normal Form (HNF) basis} for a full rank lattice $L \subseteq \R^n$ is a square, non-singular, upper triangular matrix $B \subseteq \R^{n \times n}$ such that off-diagonal elements satisfy : $0 \leq B_{i,j} < B_{j,j}$ for all $1 \leq i  < j \leq n$.

\begin{fact}\cite{miccLN} 
There exists an efficient algorithm which on input a set of rational vectors $B$, computes a basis for the lattice generated by B: the algorithm simply computes  the unique HNF basis of the lattice generated by $B$.
\end{fact}
%Let $L_1$ and $L_2$ be two lattices in $\mathbb{R}^n$. A  transformation $P: L_1 \rightarrow L_2$ is called an orthogonal linear transformation if it preserves the inner products. i.e. for all $x, y \in L_1, \inner{x,y} = \inner{Px, Py}$. 
%Two lattices $L_1$ and $L_2$ are said to be isomorphic, if there exists an orthogonal transformation mapping $L_1$ to $L_2$. \knote{Should orthogonal linear transformations necessarily be permutation matrices?} \gnote{ No, negating a column, rotations are also orthogonal linear transformations}
%\subsection{Construction-A of Lattices from Codes}

We note that $L_C=C+q\Z^n$ contains $q \mathbb{Z}^n$ as a sublattice and hence it is a full rank lattice.
 
 \begin{fact}
 \label{fact:basis-generator-transform}
 A basis $B$ for the lattice $L_C$ specified by the generator matrix $G$ for the code $C$ can be computed efficiently by taking the HNF of the matrix
$\begin{bmatrix} G \\ qI_n \end{bmatrix}$. Conversely, given a basis $B$ of $L_C$, the generator matrix for $C$ can be computed efficiently by finding a basis for $B \mmod q $ by row reduction over $\F_q$. 
 \end{fact}
 \begin{proof}[Proof of Fact ~\ref{fact:basis-generator-transform}]
Let $L_C$ be a lattice obtained by Construction-A from a $q$-ary linear code $C \subseteq \F_q^n$, $L_C= C+q\Z^n$. 
We first show that given a generator $G$ for the linear code $C$, the $HNF(\begin{bmatrix} G \\ qI_n \end{bmatrix})$ gives a basis for the lattice $L_C$.

Let $B = HNF(\begin{bmatrix} G \\ qI_n \end{bmatrix})$. By definition of the $HNF$ basis, $B$ is a basis for the lattice which contains each generator vector $g \in G$ and each $q e_j$ for all $j \in [n]$. %Therefore, all linear combinations of these vectors are also in $L(B)$ and it follows that $L_C \subseteq L(B)$. 
We note that each vector $v \in L_C$ is a linear combination of the generators of $C$ and $3I_n$ which is exactly the lattice $L(B)$. Therefore, $B$ is a basis for $L_C$. 

Given a basis $B$ for $L_C$, we now show that the set of linearly independent vectors in $\F_q^n$ obtained by embedding $B \mmod q$  into $\F_q$ gives a generator for the code $C$. 
 $L_C$,  contains $q \Z^n$ as a sublattice and form Claim~\ref{claim:c-is-Lmodq}, we can conclude that the code $C$  is the embedding of $L_C \mmod q$  into $\F_q$. 
Since any lattice vector $v \in L_C$ , is an integer linear combination of rows of $B$, all codewords in $L_C \mmod q$ can be obtained as linear combinations of $B \mmod q$ over $\F_q$. Therefore, the linearly independent set of vectors in $B \mmod q$ form a generator for the code $C$.
%every vector v in L_C can be mapped to a codeword in $C$ by embedding $v \mmod q$ into $\F_q$. 
\end{proof}

A {\em weighing matrix } of order $n$ and weight $k$ is a $n \times n$ matrix with entries in $\{0, 1, -1\}$ such that each row and column has exactly $k$ non-zero entries and the row vectors are orthogonal to each other.  By definition, a weighing matrix $W$ satisfies $WW^T = k I_n$. 
For matrices $A\in \R^{n_1\times n_1}$ and $B\in \R^{n_2\times n_2}$, we denote the $(n_1+n_2)\times (n_1+n_2)$-dimensional block-diagonal matrix obtained using blocks $A$ and $B$ by $A\otimes B$. We will use the following characterization of weighing matrices of weights $2$ and $3$. For completeness we present a proof of Theorems \ref{thm:weight-2-matrices} and \ref{thm:weight-3-matrices} here. %in the Section~\ref{sec:weighing-proofs}.

\begin{theorem}\cite{chan86}\label{thm:weight-2-matrices}
A matrix $W$ is a weighing matrix of order $n$ and weight $2$ if and only if $W$ can be obtained from \[\otimes_{i=1}^{n/2}\begin{bmatrix} 1 & 1\\ 1 & -1\end{bmatrix}\] by negating some rows and columns and by interchanging some rows and columns. 
\end{theorem}
\begin{proof}%[Proof of Theorem \ref{thm:weight-2-matrices} ]
Let $W(n,2)$ denote a weighing matrix of order $n$ and weight $2$. We prove this theorem by induction on the order $n$ of $W(n,2)$. 

For $n = 2$, the matrix $\begin{bmatrix} 1 & 1\\ 1 & -1\end{bmatrix}$ is the only possible $2 \times 2$ orthogonal matrix up to permutations of columns with entries in $\{ 1, -1\}$. Therefore, $W(2,2) \cong \begin{bmatrix} 1 & 1\\ 1 & -1\end{bmatrix}$. 

Let us assume the induction hypothesis about all weighing matrices of order at most $n-2$ and weight $2$. 

Let $W \in \{0 , 1, -1 \}^{n \times n}$ be an orthogonal matrix such that each row of $W$ has exactly two non-zero entires. 
Since we are characterizing $W$ up to permutations of rows and columns, and negations of rows and columns, we can assume without loss of generality that the first row of $W$ is, 
$$ w_1 = (1, 1, 0, \ldots, 0). $$
Since $W$ is orthogonal, the non-zero entries of every other row, $w_i$ has even intersection with the non-zero entries of $w_1$, i.e. $$| \supp{w_i} \cap \supp{w_1} | \in \{0, 2 \}. $$

Let us consider the case when $| \supp{w_i} \cap \supp{w_1} | = 0$ for all $i \in  [n] \setminus \{1\}$. This would imply that $W$ has at most $n-1$ rows in total which contradicts the fact that $W$ is a $n \times n$ matrix. Therefore, there exists at least one row, say $w_2$ such that $| \supp{w_2} \cap \supp{w_1} | = 2$. Since $w_2$ is orthogonal to $w_1$ and it has exactly two non-zero entries, it is of the form $$w_2 = (1, -1, 0, \ldots, 0).$$ %up to permutations of coordinates. 

We note that there cannot exist any other row $w_3$ of $W$, such that $| \supp{w_3} \cap \supp{w_1} | = 2$ since it is not possible for such a vector to be orthogonal to both $w_1$ and $w_2$. Therefore, for every other row $w_i, i \in \{ 3, \cdots, n \}$, we have  $| \supp{w_i} \cap \supp{w_1} | = 0$. The weighing matrix is therefore of the form: $$ W \cong \begin{bmatrix} 1 & 1\\ 1 & -1\end{bmatrix} \otimes W' $$ where $W'$ is a weighing matrices of order at $n-2$ and weight $2$. The proof follows from the induction hypothesis. 
\end{proof}

\begin{theorem}\cite{chan86}\label{thm:weight-3-matrices}
A matrix $W$ is a weighing matrix of order $n$ and weight $3$ if and only if $W$ can be obtained from $\otimes_{i=1}^{n/4}M$ by negating some rows and columns and by interchanging some rows and columns. 
\end{theorem}
\begin{proof}%[Proof of Theorem \ref{thm:weight-3-matrices}]
Let $W(n,3)$ denote a weighing matrix of order $n$ and weight $3$. We prove this theorem by induction on the order $n$ of $W(n,3)$. 

For $n = 4$, from Lemma~\ref{lemma:w(4,3)} we have $W(4,3) \cong M$. 
Let us assume the induction hypothesis about all weighing matrices of order at most $n-4$ and weight $3$. 

Let $W \in \{0 , 1, -1 \}^{n \times n}$ be an orthogonal matrix such that each row of $W$ has exactly three non-zero entires. 
Since we are characterizing $W$ up to permutations of rows and columns, and negations of rows and columns, we can assume without loss of generality that the first row of $W$ is 
$$ w_1 = (1, 1, 1, 0, \ldots, 0).$$
Since $W$ is orthogonal, the non-zero entries of every other row, $w_i$ has even intersection with the non-zero entries of $w_1$, i.e. $$| \supp{w_i} \cap \supp{w_1} | \in \{0, 2 \}~ \text{ for all } i \in \{2, \ldots, n\}. $$

Let us consider the case when $| \supp{w_i} \cap \supp{w_1} | = 0$ for all $i \in  [n] \setminus \{1\}$. This would imply that $W$ has at most $n-2$ rows in total which contradicts the fact that $W$ is a $n \times n$ matrix. Therefore, there exists at least two rows, say $w_2, w_3 $ such that $| \supp{w_1} \cap \supp{w_2} | = 2$ and $| \supp{w_1} \cap \supp{w_3} | = 2$. 
Since these three vectors are mutually orthogonal and $\supp{w_1} = 2$, it follows that $| \supp{w_2} \cap \supp{w_3} | > 0$. Without loss of generality, these three vectors are of the following form:
$$\begin{bmatrix} 1&1&1&0 & 0&\cdots &0 \\
1& -1& 0& 1& 0& \cdots & 0\\
1 & 0 & -1& -1 & 0& \cdots & 0
%\vdots & &&&\cdots& &\vdots
\end{bmatrix}
$$
We observe that if $| \supp{w_i} \cap \supp{w_2} | = 0$ for all $i \in  [n] \setminus \{1, 3\}$, then the number of vectors in $W$ is at most $n-1$. Therefore, 
there exists at least one other row $w_4$ such that $| \supp{w_4} \ \cap \ \supp{w_2} | = 2$. Since $w_4$ is orthogonal to all $w_1, w_2$ and $w_3$,
%other rows, 
the unique candidate for $w_4$ is of the form $(0, 1, -1, 1, 0, \ldots, 0)$. We note that if there exists another row, $w_5$ such that $| \supp{w_5} \ \cap \ \supp{w_j} | > 0 $ for any $j \in [4]$ , then it cannot be orthogonal to all four vectors $w_1, w_2, w_3$ and $w_4$. So, $| \supp{w_i} \ \cap \ \supp{w_j} | = 0 $ for any $j \in [4]$ and $i \geq 5$.

Therefore, $W(n,3) \cong M \otimes W'$, where $W'$ is a weighing matrix of order $n-4$ and weight $3$. It then follows from the induction hypothesis that 
$$W(n,3) \cong \otimes_i M. $$
\end{proof}

\begin{lemma}\label{lemma:w(4,3)}
A weighing matrix of order $4$ and weight $3$ is equivalent to $M$ up to permutations of rows and columns, and negations of rows and columns.
\end{lemma}
\begin{proof}
Let $W$ be a weighing matrix of order $4$ and weight $3$. Since each vector has weight at exactly $3$, we can assume without loss of generality, $w_1 = (1,1,1, 0)$. 
All rows are mutually orthogonal, therefore, $| \supp{w_i} \ \cap \ \supp{w_j} | \in \{0, 2\}$ and $| \supp{w_1} \ \cap \ \supp{w_i} | \neq 0$ for all $i$. 

Let us consider another row $w_2$ such that $ | \supp{w_1} \ \cap \ \supp{w_2} | = 2$.  So, $w_2 = (1, -1, 0, 1)$ up to permutations of coordinates. For any other row $w_3$, if $| \supp{w_1} \ \cap \ \supp{w_2} \ \cap \ \supp{w_3} | = 2$ then the orthogonality condition with either $w_1$ or $w_2$ is violated. Therefore, $w_3$ is of the form $w_3 = (1, 0, -1, -1)$. This forces $w_4 = (0, 1, -1, 1)$ and, hence $W \equiv M$. %$$W(4,3) \cong M$ follows. 
\end{proof}

\section{Orthogonal Lattices from Binary Codes}\label{sec: binary}
In this section we focus on lattices obtained from binary linear codes using Construction-A. In Section \ref{sec:decomp-2Z}, we show that any orthogonal lattice obtained from a binary linear code by Construction-A is equivalent to a product lattice whose components are one-dimensional or two-dimensional lattices. In Section \ref{sec:algo-2Z}, we show that given a lattice obtained from a binary linear code by Construction-A, there exists an efficient algorithm to verify if the lattice is orthogonal. 

\subsection{Decomposition Characterization}\label{sec:decomp-2Z}
We prove Theorem \ref{thm:binary-decomposition} in this subsection. 
%We will use the following lemma about product lattices generated from codes. 

\begin{proof}[Proof of Theorem \ref{thm:binary-decomposition}]
We show that $(1) \equiv (2)$ and $(2) \equiv (3)$ to complete the equivalence of the three statements.

\noindent $(1) \equiv (2)$: We show that $L_C = C+ 2\Z^n$  is orthogonal if and only if it decomposes into direct product of lower dimensional orthogonal lattices, $L_C  \cong \otimes_i L_i$.

If $L_C  \cong \otimes_i L_i$ such that each $L_i$ is orthogonal, then $L_C$ is also orthogonal. This is because $L_C$ would have a block diagonal orthogonal basis where each block is in itself orthogonal or a $1 \times 1$ matrix. 

We prove the other direction of the equivalence by induction on the dimension, $n$, of the lattice $L_C$. 
For the base case consider $n=1$. Since $L$ is integral, contains $2\Z$ and is of the form $C+2\Z$ for some binary linear code $C$, it follows that $L$ has to be either $\Z$ or $2\Z$. 

Let us assume the induction hypothesis for all $n-1$ or lower dimensional orthogonal lattices obtained from binary linear codes using Construction-A.\\

Let $L_C$ be an $n$-dimensional orthogonal lattice and $B$ be an orthogonal basis of $L_C$ with the rows being basis vectors. Since $L_C$ is an integral lattice, $B$ has only integral entries. The next two claims summarize certain properties of the entries of the basis matrix $B$.

\begin{claim}\label{claim:property1-B-2Z}
For every row $b$ of $B$ and for every $j\in [n]$, we have that $2| b_j |\in \{0,\|b\|^2,2\|b\|^2\}$.
\end{claim}
\begin{proof}
Since $B$ is an orthogonal basis, $BB^T = D$, where $D$ is the diagonal matrix with $d_i = \| b^{(i)} \|^2$, where $b^{(i)}$ denotes the $i^{th}$ basis vector.  
$$
D=\left[
\begin{array}{ccccc}
\|b^{(1)}\|^2 & 0 & 0 & \cdots & 0\\ 
0 & \|b^{(2)}\|^2 & 0 & \cdots & 0 \\
\vdots& \vdots & \ddots &  & \vdots\\
0 & 0 & 0 & \cdots & \|b^{(n)}\|^2 \\
\end{array}\right]
$$

We know that $2\mathbb{Z}^n \subseteq L_C$ so, $2e_j \in L_C$ for every $j \in [n]$. Therefore, there is an integral matrix $X \in \mathbb{Z}^{n \times n}$ such that   $XB  =  2I_n$, i.e.  $2  B^{-1} \in \mathbb{Z}^{n \times n}$. Since $B$ is an orthogonal basis,
\[ B^{-1} = B^T D^{-1} \in \frac12\mathbb{Z}^{n \times n}. \]

Each column of $B^T D^{-1}$ is given by $b / \|b\|^2$, where $b$ is a basis vector. Therefore, for any $j \in [n], 2b_j$ is a multiple of $\|b\|^2$.  Thus we have 
\begin{equation*} %\label{eq:mod-condition}
2b_j \equiv 0 \text{ mod } \|b\|^2 \text{ for all } j\in [n], \text{ and rows $b$ of $B$}. 
\end{equation*} 

Since $b_j$ is integral and $ \lvert b_j \rvert \leq \|b\|^2$ for every $j\in [n]$,  it follows that $2|b_j|\in\{0, \|b\|^2, 2\|b\|^2\}$. 
%Suppose there exists $j\in[n]$ such that $3|b_j|=2\|b\|^2$. Since $b$ is a basis vector, it follows that $b$ is not all zeroes. Hence $b_j\neq 0$. We can re-write the condition $3 \lvert b_j \rvert  = 2\|b\|^2$ as $3\lvert b_j \rvert  = 2\sum_{i=1}^n b_i^2$. Rearranging the terms, we have
%\[
%\lvert b_j \rvert~ (3 - 2 \lvert b_j \rvert) = 2\sum_{i \neq j} b_i^2.
%\]
%Since the RHS is a sum of squares, it is always non-negative. The LHS is non-zero since $b_j \in \Z\setminus\{0\}$. So the LHS should be strictly positive. Therefore, $|b_j|\in (0,3/2)\cap \Z$ and hence $|b_j|=1$. However, this implies that $\sum_{i \neq j} b_i^2 = 1/2$, contradicting the fact that $b$ is integral. Hence, $3 \lvert |b_j| \rvert  = 2\|b\|^2$ is impossible. 
\end{proof}

\begin{claim}\label{claim:properties-of-B-2Z}
Let $b$ be a row of $B$.
\begin{enumerate}
\item[(1)] If there exists $j\in [n]$ such that $2|b_j |=2\|b\|^2$, then $b_j=\pm 1$ and $b_{j'}=0$ for every $j'\in [n]\setminus \{j\}$.
\item[(2)]  If there exists $j\in [n]$ such that $2|b_j|=\|b\|^2$ and $b_j=\pm 2$, then $b_{j'}=0$ for every $j'\in [n]\setminus \{j\}$.
\item[(3)]  If there exists $j\in [n]$ such that $2|b_j |=\|b\|^2$ and $b_j=\pm 1$, then there exist ${j_1} \in [n]\setminus\{j\}$, such that $\lvert b_{j_1} \rvert = 1$ and $b_{j'}=0$ for every $j'\in [n]\setminus\{j,j_1\}$.
\end{enumerate}
\end{claim}
\begin{proof}

\begin{enumerate}
\item[(1)]  Since,  $\| b\|^2 = \sum_{i=1}^n b_i ^2$, and each $b_i \in \mathbb{Z}$, we conclude that $\lvert b_j \rvert = 1$ and the remaining coordinates in $b$ have to be $0$, i.e $b_{j'} = 0$ for all $j'\in[n]\setminus\{j\}$.
\item[(2)]   Follows from $2|b_j|=\|b\|^2$ and $b$ being integral. 
\item[(3)]   We can re-write the condition $2 \lvert b_j \rvert  = \|b\|^2$ as $2\lvert b_j \rvert  = \sum_{i=1}^n b_i^2$. Rearranging the terms, we have
\begin{equation*}%\label{eq:1-2Z}
\lvert b_j \rvert ~(2 -  \lvert b_j \rvert) = \sum_{i \neq j} b_i^2. 
\end{equation*}
If $b_j =\pm 1$, then $\sum_{i \neq j} b_i^2 = 1$. Further, $b$ is integral. Hence, $b$ has exactly $1$ other non-zero coordinates $b_{j_1}$, $j\neq j_1$, such that $\lvert b_{j_1} \rvert = 1$. 
\end{enumerate}
\end{proof}

Using the properties of the orthogonal basis $B$ of $L_C$ given in Claims \ref{claim:property1-B-2Z} and \ref{claim:properties-of-B-2Z}, we show that $B$ is equivalent (up to permutations of its columns) to a block diagonal matrix, i.e
\[ B \cong \begin{bmatrix} B_1 & 0 & \cdots & 0 \\
0& B_2 & \cdots & 0 \\
\vdots& & \ddots & \vdots \\
0 & 0 & \cdots & B_k
\end{bmatrix}\]
where each $B_i$ is either the $1\times 1$ matrix $\begin{bmatrix} 1 \end{bmatrix}$ or the $1\times 1$ matrix $\begin{bmatrix}2\end{bmatrix}$ or the $2\times 2$ matrix $\begin{bmatrix}1&1\\1&-1 \end{bmatrix}$. It follows that $L_C \cong \otimes_i L_i$ such that $B_i$ is the basis for the lower dimensional lattice $L_i$.

Let us pick a row $b$ of $B$ with the smallest support. %and which has maximum $\|b\|_1$. 
Fix an index $j \in [n]$ to be the index of a non-zero entry with minimum absolute value in $b$, i.e. $j := \arg \min_k \{ \lvert b_k \rvert \}$. Since $b$ is a row of a basis matrix, $b$ cannot be the all-zeroes vector and therefore there exists a $j \in [n]$ such that $|b_j|>0$. Since we are only interested in equivalence (that allows for permutation of coordinates), we may assume without loss of generality that $j=1$ by permuting the coordinates. By Claim \ref{claim:property1-B-2Z}, we have that $2 \lvert b_1 \rvert \in \{ \|b\|^2,2\|b\|^2 \}$. We consider each of these cases separately. 

\begin{enumerate}[leftmargin=*]
\item Suppose $2 \lvert b_1 \rvert  = 2 \|b\|^2$. %Since,  $\| b\|^2 = \sum_{i=1}^n b[i] ^2$, and each $b[i] \in \mathbb{Z}$, we conclude that $\lvert b[j] \rvert = 1$ and the remaining coordinates in $b$ have to be $0$, i.e $b[i] = 0$ for all $i \neq j \in [n]$. 
%So, 
By Claim \ref{claim:properties-of-B-2Z}(1), $b = (\pm 1, 0,  \ldots, 0 )$. Since $B$ is an orthogonal basis, $\inner{b, b'} = 0 \Rightarrow b'_1 = 0$ for all $b' \neq b \in B$. The orthogonality of $B$ therefore forces all other basis vectors to take a value of $0$ in the $1$'st coordinate. Thus $B$ is of the form

\[B = 
\left(
\begin{array}{c|c}
  \pm1 & 0 \cdots 0 \\ \hline
  0 & \raisebox{-15pt}{{\Large\mbox{{$B'$}}}} \\[-4ex]
  \vdots & \\[-0.5ex]
  0 &
\end{array}
\right).
\]

Therefore, we obtain $L_C \cong \Z \otimes L'$, where $L'$ is an orthogonal $(n-1)$-dimensional lattice generated by the basis matrix restricted to the coordinates other than $1$, say, $B'$.  
From Claim~\ref{lemma:product_codes_lattices}, it follows that $L' = C' + 2 \Z^{n-1}$ for some ternary linear code $C' \subseteq \mathbb{F}_2^{n-1}$. Thus $L'$ satisfies the induction hypothesis and we have the desired decomposition.

%Permuting the columns of a basis is an orthogonal transformation and therefore, we can equivalently write $b \equiv (1, 0, 0, \cdots, 0)$. So, $L_C \cong \mathbb{Z} \times L' $. 
%Since $b \mod 3 \in C$, the code is also of the form $C \cong C_1 \times C'$,where $C'$ is the projection of $C$ obtained by omitting the first column of $C$ and $C_1 =  \mathbb{F}_3^1$. 

\item Suppose $2 \lvert b_1 \rvert  = \|b\|^2$. We can re-write this condition as $2\lvert b_1 \rvert  = \sum_{i=1}^n b_i^2$. Rearranging the terms, we have
\[
\lvert b_1 \rvert ~(2 -  \lvert b_1 \rvert) = \sum_{i \neq 1} b_i^2. 
\]
Since the RHS is a sum of squares, it should be non-negative.

\begin{description}[leftmargin=*, listparindent=\parindent]
\item[(i)] If RHS is $0$, then $b_1 = \pm  2$ and therefore, it follows from Claim~\ref{claim:properties-of-B-2Z}(2) that $b = (\pm 2, 0,  \ldots, 0)$. The orthogonality of $B$ forces all other basis vectors to take a value of $0$ at the $1$'st coordinate. 
\[B = 
\left(
\begin{array}{c|c}
  \pm2 & 0 \cdots 0 \\ \hline
  0 & \raisebox{-15pt}{{\Large\mbox{{$B'$}}}} \\[-4ex]
  \vdots & \\[-0.5ex]
  0 &
\end{array}
\right)
\]
Therefore, we obtain $L_C\cong 2\Z \otimes L'$, where $L'$ is an orthogonal $(n-1)$-dimensional lattice generated by the basis matrix restricted to the coordinates other than $1$, say $B'$. 
From Claim~\ref{lemma:product_codes_lattices}, it follows that $L' = C' + 2 \Z^{n-1}$ for some ternary linear code $C' \subseteq \mathbb{F}_2^{n-1}$. Thus $L'$ satisfies the induction hypothesis and we have the desired decomposition. 
				
\item[(ii)] If RHS is strictly positive, then $|b_1|\in (0,2)\cap \Z=\{1\}$. By Claim \ref{claim:properties-of-B-2Z}(3), we have that $b$ has exactly two non-zero coordinates and they are $\pm 1$. By permuting the coordinates of $B$, we may assume that $b \equiv (\pm 1, \pm 1, 0, \cdots, 0)$. 		 

Since we picked the row $b$ to be the one with the smallest support, it follows that every row has at least $2$ non-zero coordinates. %By the arguments in (a), 
By Claims \ref{claim:property1-B-2Z} and \ref{claim:properties-of-B-2Z}, this is possible only if 
for every other row $b'$, there exists $j'\in[n]$ such that $2|b'_{j'} |=\|b'\|^2$. By Claim \ref{claim:properties-of-B-2Z}(1) and (2), every other row $b'$ has all its coordinates in $\{0,\pm 1, \pm 2\}$. 
%By the arguments in (i), 
By Claim \ref{claim:properties-of-B-2Z}(2), every other row $b'$ has none of its coordinates in $\{\pm 2\}$. Therefore, every other row $b'$ has all its coordinates in $\{0,\pm 1\}$. By Claim \ref{claim:properties-of-B-2Z}(3), every row of the basis matrix has the same form as $b$: they have exactly two non-zero entries each of which is $\pm 1$. 

Since the rows of the basis matrix are orthogonal, it follows that the basis matrix $B$ is a weighing matrix of order $n$ with weight $2$. By Theorem \ref{thm:weight-2-matrices} the matrix $B$ is obtained from $\otimes_{n/2} \begin{bmatrix}1&1\\1&-1 \end{bmatrix}$ by either negating some rows or columns and by interchanging rows or columns. We recall that interchanging or negating the rows of the basis matrix of a lattice preserves the basis property while interchanging columns is equivalent to permuting the coordinates. Hence $L_C = L(B)\cong \otimes_{i=1}^{n/2}L(\begin{bmatrix}1&1\\1&-1 \end{bmatrix})$.
%there exists a signed permutation matrix $U$ such that $UB=\otimes_{n/4}M$. Since signed permutation matrices are unimodular matrices, the lattice $L$ is generated by the basis matrix $\otimes_{n/4}M$.
\end{description}
\end{enumerate}

\noindent $(2) \equiv (3)$: We now show that $L_C$ decomposes into direct product of lower dimensional lattices, $L_C \cong \otimes_i L_i$ if and only if the code $C$ also decomposes, $C \cong \otimes_i C_i$.

Let $L_C\cong \otimes_i L_i$. Without loss of generality, we can consider $L_C = \otimes_i L_i$. We have $C = L_C \mod 2= \otimes_i L_i \mmod 2$. 
We observe that if $L_i$ has dimension $n_i$, then $L_i\supseteq 2\Z^{n_i}$. Therefore, $C_i=L_i\mmod 2$ is a binary code. 
Let $C_i := L_i \mmod 2$ for every $i$. Then $C=\otimes_i C_i$. (If $c\in C$, then $c\in L$ and hence the projection of $c$ to the subset of coordinates corresponding to $L_i$ is in $C_i$. Let $c_i\in C_i$ for every $i$. The concatenated vector $\otimes_i c_i$ is in $\otimes_i L_i \mod 2$ and hence is in $C$).%Therefore, we get $C = \otimes_i C_i$. 

To show the other direction of the equivalence, let $C \cong \otimes_i C_i$, where each $C_i \subseteq \mathbb{F}_2^{n_i}$ and $n = \sum_i n_i$. Therefore $L_C = C+ 2\mathbb{Z}^n  \cong \otimes_i C_i + 2\Z^n \cong \otimes_i  (C_i + 2 \mathbb{Z}^{n_i})$, since $\mathbb{Z}^n \cong \otimes_i \mathbb{Z}^{n_i}$.

%\end{enumerate}
\end{proof}

%%%%%%%%%%%%%%%%%%%%%%%%%%%%%%%%%%%%%%%%%%%%%%%%%%%%%%%%%%%%%%%%%%%%%%%%%%%%%%%%
%ALGORITHM

\subsection{Algorithm}\label{sec:algo-2Z}

Theorem~\ref{thm:binary-decomposition} shows that a lattice of the form $C+ 2\mathbb{Z}^n$ is orthogonal if and only if the underlying code decomposes into direct product of binary linear codes isomorphic to $\{ 0, 1\}$ or $\{0\}$ or the two dimensional code $\{00, 11 \}$. We now give a polynomial time algorithm which finds the decomposition of the code $C$ into the component codes, $C_i$, if there exists one. Therefore, if the lattice $L_C$ is orthogonal, the algorithm decides in polynomial time if it is orthogonal and also gives the orthogonal basis for the lattice. 

The algorithm recursively attempts to find the component codes. If it is unable to decompose the code at any stage, then it declares that $L_C$ is not orthogonal. At every step we check if $C \cong \{0,1\} \times C'$ or $\{0\}\times C'$ or  $\{ 00, 11 \} \times C'$ and then recurse on $C'$.

\begin{proof}[Proof of Theorem~\ref{thm:algo-binary}]
Given a basis for $L_C$ as input, we first compute the generator for $C$. From Theorem~\ref{thm:binary-decomposition}, we know that if $L_C$ is orthogonal, then $C \cong \otimes_i C_i$ where each $C_i$ is either the length-1 code $\{0, 1\}$ or the length-1 code $\{0\}$ or the $2$-dimensional code \{00, 11 \}. 

The algorithm therefore in each step decides if $C \cong \{0,1\} \otimes C'$ or $C\cong \{0\}\otimes C'$ or $C \cong \{ 00, 11\} \otimes C'$.
Theorem~\ref{theorem:code-decomposition-01-2Z} shows that using Algorithm~\ref{algo1-2Z} we can check in $O(n^4)$ time, if $C \cong \{0,1\} \otimes C'$. The same algorithm can be modified to check in $O(n^4)$ time, if $C\cong \{0\}\otimes C'$. Theorem~\ref{theorem:code-decomposition-Cp-2Z} shows that Algorithm~\ref{algo2-2Z} can verify if $C \cong \{00, 11\} \otimes C'$ in $O(n^5)$ time. If any one of the algorithms finds a decomposition, then we recurse on the lower dimensional code $C'$ to find a further decomposition. We recurse at most $n$ times. If all the algorithms fail to find a decomposition, then $L_C$ is not orthogonal. Therefore, it takes $O(n^6)$ time to decide if $L_C$ is orthogonal. 
\end{proof}

We now describe the individual algorithms to verify if $C \cong \{0,1\} \otimes C'$ or $C\cong \{0\}\otimes C'$ or $C \cong \{00, 11\} \otimes C'$. 

\begin{algorithm}
  \caption{$\mathbf{: decompose-length-1(G)}$:}
  \label{algo1-2Z}
  \textbf{Input}: $G=\{g_1,\ldots, g_n\} \in \F_2^{n}$ (A generator for the code $C$)\\ 

  \begin{algorithmic}[1]
  \FOR {$j \in \{ 1, \cdots, n \} $ }
  	\STATE Let $G' \leftarrow$ projection of vectors in $G$ on coordinates $[n] \setminus \{j\}$
  	\STATE For $g\in G'$, define $g^0,g^1 \in \F_2^n$ as the $n$-dimensional vectors obtained by extending $g$ using $0$ and $1$ along the $j$'th coordinate respectively. 
	\IF{ $g^0,\ g^1 \in C$ for all $g \in G'$}
		\STATE \textbf{return } $j$
	\ENDIF
  \ENDFOR
  \STATE \textbf{return } FAIL
  \end{algorithmic}
\end{algorithm}
 
\begin{theorem}
\label{theorem:code-decomposition-01-2Z}
Let $C$ be a binary linear code and $G=\{g_1,\ldots, g_n\}\in \F_2^{n\times n}$ be its generator. Then Algorithm \ref{algo1-2Z} decides if $C\cong \{0,1\}\otimes C'$ for some linear code $C'\subseteq \F_2^{n-1}$ and if so outputs the coordinate corresponding to the direct product decomposition. Moreover the algorithm runs in time $O(n^4)$.

\end{theorem}
\begin{proof}
For $j\in [n]$, let $C_{\overline{j}}'\subseteq \F_2^{n-1}$ be the projection of $C$ on the indices $[n] \setminus \{j\}$ and for a vector $c\in C_{\overline{j}}'$, let $c^0,c^1 \in \F_2^n$ be extensions of $c$ using $0,1$ respectively along the $j$'th coordinate. We note that $C \cong \{0,1\} \otimes C'$ for some binary linear code $C'$ if and only if there exists an index  $j \in [n]$, such that  
\begin{equation*} \label{check1} 
C = \left\{ c^0, c^1 \mid c \in C_{\overline{j}}' \right\}. 
\end{equation*}
From the definition of $C_{\overline{j}}'$, it follows that $C\subseteq \{ c^0, c^1 \mid  c \in C_{\overline{j}'} \}$ up to a permutation of coordinates. So, the algorithm just needs to verify if the other side of the containment holds for some $j \in [n]$.

Let $G'$ be the set of vectors of $G$ projected on the coordinates $[n] \setminus \{j\}$. Algorithm~\ref{algo1-2Z} verifies if $g^0$ and $g^1$ are codewords in $C$, for every vector $g \in G'$. We now show that this is sufficient.
Since $C$ is a code, if $g^0$, $g^1$ $\in C$ for every $g\in G'$, then all linear combinations of these vectors are also in $C$.  Therefore, $\{ c^0, c^1 \mid c \in C_{\overline{j}}' \} \subseteq C$.

It takes $O(n^2)$ time to compute a parity check matrix from the generator $G$ and $O(n^2)$ time to verify if an input vector is a codeword using the parity check matrix.  
For every possible choice of the index $j$, Algorithm~\ref{algo1-2Z} checks if each of the $2n$ vectors of the form $g^0,g^1$ are in $C$. Therefore, Algorithm~\ref{algo1-2Z} takes $O(n^4)$ time to decide if $C \cong \{0, 1\} \otimes C'$. 
\end{proof}

\begin{algorithm}
  \caption{ $\mathbf{: decompose-length-2(G)}$:}
  \label{algo2-2Z}
  \textbf{Input}: $G=\{g_1,\ldots, g_n\} \in \F_2^{n}$ (A generator for the code $C$)\\

  \begin{algorithmic}[1]
  \FOR{ $j_1, j_2 \in \{1, 2, \cdots, n \} $}
	\STATE Let $G' \leftarrow$ projection of vectors in $G$ on coordinates $[n] \setminus \{ j_1, j_2 \}$
	\STATE Let $G'' \leftarrow$ projection of vectors in $G$ on coordinates $\{ j_1, j_2 \}$
  	\IF{ Code generated by $G'' \equiv \{ 00, 11\} $}
		\STATE For $g \in G'$ define $g^{00}, g^{11} \in \F_2^n$ be $n$-dimensional vectors obtained by extending $g$  using $00$ and $11$ along the $j_1, j_2$ coordinates. 
		\IF{ $g^{00}, g^{11} \in C$ for all $g \in G'$}
			\STATE \textbf{ return } $j_1, j_2$
		\ENDIF
	\ENDIF
  \ENDFOR
  \STATE \textbf{ return } FAIL
  \end{algorithmic}
\end{algorithm}

\begin{theorem}
\label{theorem:code-decomposition-Cp-2Z}
Let $C$ be a binary linear code and $G=\{g_1,\ldots, g_n\}\in \F_2^{n\times n}$ be its generator. 
Then Algorithm \ref{algo2-2Z} decides if $C \cong \{00, 11\} \otimes C'$ for some linear codes $C'\subseteq \F_2^{n-2}$ and if so outputs the coordinates corresponding to the direct product decomposition. Moreover the algorithm runs in time $O(n^5)$.
\end{theorem}
\begin{proof}
For $j_1, j_2 \in [n]$, let $C''_{j_1, j_2}$ be the projection of $C$ on the indices $\{j_1, j_2\}$. We first verify if $C''_{ j_1, j_2}$ is the code \{00, 11\}. For this purpose, it is sufficient to check if $C''_{j_1, j_2}$ is generated by $\{11\}$. 
Now, to see if $C \cong \{ 00, 11 \} \otimes C'$ for some binary linear code $C' \subseteq \F_2^{n-2}$. Define $C'_{\bar{j_1}, \bar{j_2} }$ to be the projection of $C$ on the indices $[n] \setminus \{j_1, j_2 \}$. For a vector $c \in C'_{\bar{j_1}, \bar{j_2}}$, let $c^{00}, c^{11} \in \F_2^n$ be the extensions of $c$ using $\{00, 11\}$ along the $j_1, j_2$ coordinates. 
We note that $C \cong \{ 00, 11 \} \otimes C'$ for some binary linear code $C'$ if and only if there exist indices  $j_1, j_2 \in [n]$, such that  
\begin{equation} \label{check2} 
C = \left\{ c^{00}, c^{11} \mid c \in C_{\bar{j_1}, \bar{j_2}}' \right\}. 
\end{equation}
From the definition of $C_{\bar{j_1}, \bar{j_2}}' $ and $C''_{j_1, j_2} = \{00, 11\}$, it follows that $C\subseteq \{ c^{00}, c^{11} \mid  c \in C_{\bar{j_1}, \bar{j_2}}' \}$. So, the algorithm just needs to verify if the other side of the containment holds for some indices $j_1, j_2 \in [n]$.

Let $G'$ be the set of vectors of $G$ projected on the coordinates $[n] \setminus \{j_1, j_2\}$. Algorithm~\ref{algo2-2Z} verifies if $g^{00}$ and $g^{11}$ are codewords in $C$, for every vector $g \in G'$. We now show that this is sufficient.
Since $C$ is a code, if $g^{00}, g^{11} \in C$ for every $g\in G'$, then all linear combinations of these vectors are also in $C$.  Therefore, $\{ c^{00}, c^{11} \mid~\forall~ c \in C_{\bar{j_1}, \bar{j_2}}' \} \subseteq C$. 

For each choice of $\{ j_1, j_2 \}$, it takes $O(n)$ time to verify if $C_{j_1, j_2}'' = \{ 00, 11\}$. Time to verify if an input vector is a codeword using the parity check matrix is $O(n^2)$.  We perform this check for $2n$ vectors of the form $\{ g^{00}, g^{11} \mid g  \in G' \}$. 

It takes $O(n^3)$ time to verify if  $C \cong \{00, 11\} \otimes C_{\bar{j_1}, \bar{j_2}}'$ for every pair of indices $j_1, j_2 \in [n]$. There are at most $n \choose 2$ possible choices of indices, $j_1, j_2$, therefore, it takes $O(n^5)$ time in total to decide if $C \cong \{00, 11\} \otimes C'$.

\end{proof}

\section{Orthogonal Lattices from Ternary Codes}\label{sec:ternary}
In this section we focus on lattices obtained from ternary linear codes using Construction-A. In Section \ref{sec:decomp}, we show that any orthogonal lattice obtained from a ternary linear code by Construction-A is equivalent to a product lattice whose components are one-dimensional or four-dimensional. In Section \ref{sec:algo}, we show that given a lattice obtained from a ternary linear code by Construction-A, there exists an efficient algorithm to verify if the lattice is orthogonal. 

\subsection{Decomposition Characterization}\label{sec:decomp}
We prove Theorem \ref{thm:ternary-decomposition} in this subsection. 

\begin{proof}[Proof of Theorem \ref{thm:ternary-decomposition}]
We show that $(1) \equiv (2)$ and $(2) \equiv (3)$ to complete the equivalence of the three statements.\\
%\begin{enumerate}
%\item[$(2) \equiv (3)$]:
%\item[$(1) \equiv (2)$]:
\noindent $(1) \equiv (2)$: We show that $L_C = C+ 3\Z^n$  is orthogonal if and only if it decomposes into direct product of lower dimensional orthogonal lattices, $L_C  \cong \otimes_i L_i$.

If $L_C  \cong \otimes_i L_i$ such that each $L_i$ is orthogonal, then $L_C$ is also orthogonal. This is because $L_C$ has a block diagonal basis where each block is itself an orthogonal matrix (by definition, a $1\times 1$-dimensional matrix is orthogonal).

We prove the other direction of the equivalence by induction on the dimension, $n$, of the lattice $L_C$. 
For the base case consider $n=1$. Since $L$ is integral, contains $3\Z$ and is of the form $C+3\Z$ for some ternary code $C$, it follows that $L$ has to be either $\Z$ or $3\Z$. 

Let us assume the induction hypothesis for all $n-1$ or lower dimensional orthogonal lattices obtained from ternary linear codes using Construction-A.

Let $L_C$ be an $n$-dimensional orthogonal lattice and $B$ be an orthogonal basis of $L_C$ with the rows being basis vectors. Since $L_C$ is an integral lattice, $B$ has only integral entries. The next two claims summarize certain properties of the entries of the basis matrix $B$.

\begin{claim}\label{claim:property1-B}
For every row $b$ of $B$ and for every $j\in [n]$, we have that $3| b_j |\in \{0,\|b\|^2,3\|b\|^2\}$.
\end{claim}
\begin{proof}
Since $B$ is an orthogonal basis, $BB^T = D$, where $D$ is the diagonal matrix with $d_i = \| b^{(i)} \|^2$, where $b^{(i)}$ denotes the $i^{th}$ basis vector.  %\gnote{need to fix this notation. }

$$
D=\left[
\begin{array}{ccccc}
\|b^{(1)}\|^2 & 0 & 0 & \cdots & 0\\ 
0 & \|b^{(2)}\|^2 & 0 & \cdots & 0 \\
\vdots& \vdots & \ddots &  & \vdots\\
0 & 0 & 0 & \cdots & \|b^{(n)}\|^2 \\
\end{array}\right]
$$

We know that $3\mathbb{Z}^n \subseteq L_C$ so, $3e_j \in L_C$ for every $j \in [n]$. Therefore, there is an integral matrix $X \in \mathbb{Z}^{n \times n}$ such that   $XB  =  3I_n$, i.e.  $3 B^{-1} \in \mathbb{Z}^{n \times n}$. Since $B$ is an orthogonal basis,
\[ B^{-1} = B^T D^{-1} \in \frac13\mathbb{Z}^{n \times n}. \]

Each column of $B^T D^{-1}$ is given by $b / \|b\|^2$, where $b$ is a basis vector. Therefore, for any $j \in [n], 3b_j$ is a multiple of $\|b\|^2$. Thus we have 
\begin{equation*} \label{eq:mod-condition}
3b_j \equiv 0 \text{ mod } \|b\|^2 \text{ for all } j\in [n], \text{ and rows $b$ of $B$}. 
\end{equation*} 

Since $b_j$ is integral and $ \lvert b_j \rvert \leq \|b\|^2$ for every $j\in [n]$,  it follows that $3|b_j|\in\{0, \|b\|^2, 2\|b\|^2, 3\|b\|^2 \}$. Suppose there exists $j\in[n]$ such that $3|b_j|=2\|b\|^2$. Since $b$ is a basis vector, it follows that $b$ is not all zeroes. Hence $b_j\neq 0$. We can re-write the condition $3 \lvert b_j \rvert  = 2\|b\|^2$ as $3\lvert b_j \rvert  = 2\sum_{i=1}^n b_i^2$. Rearranging the terms, we have
\[
\lvert b_j \rvert~ (3 - 2 \lvert b_j \rvert) = 2\sum_{i \neq j} b_i^2.
\]
Since the RHS is a sum of squares, it is always non-negative. The LHS is non-zero since $b_j \in \Z\setminus\{0\}$. So the LHS should be strictly positive. Therefore, $|b_j|\in (0,3/2)\cap \Z$ and hence $|b_j|=1$. However, this implies that $\sum_{i \neq j} b_i^2 = 1/2$, contradicting the fact that $b$ is integral. Hence, $3 \lvert |b_j| \rvert  = 2\|b\|^2$ is impossible. 
\end{proof}

\begin{claim}\label{claim:properties-of-B}
Let $b$ be a row of $B$.
\begin{enumerate}
\item[(1)] If there exists $j\in [n]$ such that $3|b_j |=3\|b\|^2$, then $b_j=\pm 1$ and $b_{j'}=0$ for every $j'\in [n]\setminus \{j\}$.
\item[(2)] If there exists $j\in [n]$ such that $3|b_j|=\|b\|^2$ and $b_j=\pm 3$, then $b_{j'}=0$ for every $j'\in [n]\setminus \{j\}$.
\item[(3)] If there exists $j\in [n]$ such that $3|b_j |=\|b\|^2$ and $b_j=\pm 1$, then there exist ${j_1}, {j_2}\in[n]\setminus\{j\}$, such that $\lvert b_{j_1} \rvert = \lvert b_{j_2} \rvert = 1$ and $b_{j'}=0$ for every $j'\in [n]\setminus\{j,j_1,j_2\}$.
\item[(4)] If there exists $j\in [n]$ such that $3|b_j |=\|b\|^2$, then $b_j'\in \{0,\pm 1, \pm 3\}$ for every $j'\in [n]$. 
\end{enumerate}
\end{claim}
\begin{proof}

\begin{enumerate}
\item[(1)] Since,  $\| b\|^2 = \sum_{i=1}^n b_i ^2$, and each $b_i \in \mathbb{Z}$, we conclude that $\lvert b_j \rvert = 1$ and the remaining coordinates in $b$ have to be $0$, i.e $b_{j'} = 0$ for all $j'\in[n]\setminus\{j\}$.
\item[(2)] Follows from $3|b_j|=\|b\|^2$ and $b$ being integral. 
\item[(3)] We can re-write the condition $3 \lvert b_j \rvert  = \|b\|^2$ as $3\lvert b_j \rvert  = \sum_{i=1}^n b_i^2$. Rearranging the terms, we have
\begin{equation}\label{eq:1}
\lvert b_j \rvert ~(3 -  \lvert b_j \rvert) = \sum_{i \neq j} b_i^2. 
\end{equation}

If $b_j =\pm 1$, then $\sum_{i \neq j} b_i^2 = 2$. Further, $b$ is integral. Hence, $b$ has exactly $2$ other non-zero coordinates $b_{j_1}, b_{j_2}$, $j\neq j_1, j_2$, such that $\lvert b_{j_1} \rvert = \lvert b_{j_2} \rvert = 1$. 
\item[(4)] We have equation (\ref{eq:1}). The RHS is a sum of squares and hence the LHS is non-negative. Moreover, $b$ is not all-zeroes vector implies that $b_j \neq 0$. Therefore, $|b_j|\in (0,3]\cap \Z$. If $b_j=\pm 2$, then in order to satisfy $\sum_{i\neq j}b_i^2=2$ using integral $b_i$'s, exactly two coordinates $b_{j_1},b_{j_2}$ should be $\pm 1$, where $j\neq j_1,j_2$. However, in this case, $3|b_{j_1}|= 3|b_{j_2}|=3\not\in \{0,\|b\|^2=6, 3\|b\|^2=18\}$, thus contradicting Claim \ref{claim:property1-B}.
The conclusion follows from parts (2) and (3).
\end{enumerate}
\end{proof}

Using the properties of the orthogonal basis $B$ of $L_C$ given in Claims \ref{claim:property1-B} and \ref{claim:properties-of-B}, we show that $B$ is equivalent (up to permutations of its columns) to a block diagonal matrix, i.e
\[ B \cong \begin{bmatrix} B_1 & 0 & \cdots & 0 \\
0& B_2 & \cdots & 0 \\
\vdots& & \ddots & 0 \\
0 & 0 & \cdots & B_k
\end{bmatrix}\]
where each $B_i$ is either the $1\times 1$ matrix $\begin{bmatrix}1\end{bmatrix}$ or the $1\times 1$ matrix $\begin{bmatrix}3\end{bmatrix}$ or the $4\times 4$ matrix obtained from $M$ by negating a subset of its columns, ${\cal{T}}(M)$. It follows that $L_C \cong \otimes_i L_i$ such that $B_i$ is the basis for the lower dimensional lattice $L_i$.

Let us pick a row $b$ of $B$ with the smallest support. %and which has maximum $\|b\|_1$. 
Fix an index $j \in [n]$ to be the index of a non-zero entry with minimum absolute value in $b$, i.e. $j  := \arg \min_k \{ \lvert b_k \rvert \}$. Since $b$ is a row of a basis matrix, $b$ cannot be the all-zeroes vector and therefore there exists a $j \in [n]$ such that $|b_j|>0$. Since we are only interested in equivalence (that allows for permutation of coordinates), we may assume without loss of generality that $j=1$ by permuting the coordinates. By Claim \ref{claim:property1-B}, we have that $3 \lvert b_1 \rvert \in \{ \|b\|^2,3\|b\|^2 \}$. We consider each of these cases separately. 

\begin{enumerate}[leftmargin=*]
\item Suppose $3 \lvert b_1 \rvert  = 3 \|b\|^2$. %Since,  $\| b\|^2 = \sum_{i=1}^n b[i] ^2$, and each $b[i] \in \mathbb{Z}$, we conclude that $\lvert b[j] \rvert = 1$ and the remaining coordinates in $b$ have to be $0$, i.e $b[i] = 0$ for all $i \neq j \in [n]$. 
%So, 
By Claim \ref{claim:properties-of-B}(1), $b = (\pm 1, 0,  \ldots, 0 )$. Since $B$ is an orthogonal basis, $\inner{b, b'} = 0 \Rightarrow b'_1 = 0$ for all $b' \neq b \in B$. The orthogonality of $B$ therefore forces all other basis vectors to take a value of $0$ in the $1^{st}$ coordinate. Thus $B$ is of the form

\[B = 
\left(
\begin{array}{c|c}
  \pm1 & 0 \cdots 0 \\ \hline
  0 & \raisebox{-15pt}{{\Large\mbox{{$B'$}}}} \\[-4ex]
  \vdots & \\[-0.5ex]
  0 &
\end{array}
\right).
\]

Therefore, we obtain $L_C \cong \Z \otimes L'$, where $L'$ is an orthogonal $(n-1)$-dimensional lattice generated by the basis matrix restricted to the coordinates other than $1$, say, $B'$.  
From Claim~\ref{lemma:product_codes_lattices}, it follows that $L' = C' + 3 \Z^{n-1}$ for some ternary linear code $C' \subseteq \mathbb{F}_3^{n-1}$. Thus $L'$ satisfies the induction hypothesis and we have the desired decomposition.

%Permuting the columns of a basis is an orthogonal transformation and therefore, we can equivalently write $b \equiv (1, 0, 0, \cdots, 0)$. So, $L_C \cong \mathbb{Z} \times L' $. 
%Since $b \mod 3 \in C$, the code is also of the form $C \cong C_1 \times C'$,where $C'$ is the projection of $C$ obtained by omitting the first column of $C$ and $C_1 =  \mathbb{F}_3^1$. 

\item Suppose $3 \lvert b_1 \rvert  = \|b\|^2$. We can re-write this condition as $3\lvert b_1 \rvert  = \sum_{i=1}^n b_i^2$. Rearranging the terms, we have
\[
\lvert b_1 \rvert ~(3 -  \lvert b_1 \rvert) = \sum_{i \neq 1} b_i^2. 
\]
Since the RHS is a sum of squares, it should be non-negative.

\begin{description}[leftmargin=*, listparindent=\parindent]
\item[(i)] If RHS is $0$, then $b_1 = \pm  3$ and therefore, it follows from Claim~\ref{claim:properties-of-B}(2) that $b = (\pm 3, 0,  \ldots, 0)$. The orthogonality of $B$ forces all other basis vectors to take a value of $0$ in the $1^{st}$ coordinate. 
\[B = 
\left(
\begin{array}{c|c}
  \pm3 & 0 \cdots 0 \\ \hline
  0 & \raisebox{-15pt}{{\Large\mbox{{$B'$}}}} \\[-4ex]
  \vdots & \\[-0.5ex]
  0 &
\end{array}
\right)
\]
Therefore, we obtain $L_C\cong 3\Z \otimes L'$, where $L'$ is an orthogonal $(n-1)$-dimensional lattice generated by the basis matrix restricted to the coordinates other than $1$, say $B'$. 
From Claim~\ref{lemma:product_codes_lattices}, it follows that $L' = C' + 3 \Z^{n-1}$ for some ternary linear code $C' \subseteq \mathbb{F}_3^{n-1}$. Thus $L'$ satisfies the induction hypothesis and we have the desired decomposition. 
				
\item[(ii)] If RHS is strictly positive, then $|b_1|\in (0,3)\cap \Z=\{1,2\}$. By Claim \ref{claim:properties-of-B}(4), $b_1\neq \pm 2$. Therefore, $b_1=\pm 1$. By Claim \ref{claim:properties-of-B}(3), we have that $b$ has exactly three non-zero coordinates and they are $\pm 1$. By permuting the coordinates of $B$,  we may assume that  $b \equiv (\pm 1, \pm 1, \pm 1, 0, \cdots, 0)$. 		 

Since we picked the row $b$ to be the one with the smallest support, it follows that every row has at least $3$ non-zero coordinates. %By the arguments in (a), 
By Claims \ref{claim:property1-B} and \ref{claim:properties-of-B}, this is possible only if 
for every other row $b'$, there exists $j'\in[n]$ such that $3|b'_{j'} |=\|b'\|^2$. By Claim \ref{claim:properties-of-B}(4), every other row $b'$ has all its coordinates in $\{0,\pm 1, \pm 3\}$. 
%By the arguments in (i), 
By Claim \ref{claim:properties-of-B}(2), every other row $b'$ has none of its coordinates in $\{\pm 3\}$. Therefore, every other row $b'$ has all its coordinates in $\{0,\pm 1\}$. By Claim \ref{claim:properties-of-B}(3), every row of the basis matrix has the same form as $b$: they have exactly three non-zero entries each of which is $\pm 1$. 

Since the rows of the basis matrix are orthogonal, it follows that the basis matrix $B$ is a weighing matrix of order $n$ with weight $3$. By Theorem \ref{thm:weight-3-matrices}, the matrix $B$ is obtained from $\otimes_{n/4} M$ by either negating some rows or columns and by interchanging rows or columns. We recall that interchanging or negating the rows of the basis matrix of a lattice preserves the basis property while interchanging columns is equivalent to permuting the coordinates. Hence $L_C=L(B)\cong \otimes_{i=1}^{n/4}L({\cal{T}}_i(M))$, where each ${\cal{T}}_i(M)$ is a $4\times 4$ matrix obtained by negating a subset of columns of $M$.
%there exists a signed permutation matrix $U$ such that $UB=\otimes_{n/4}M$. Since signed permutation matrices are unimodular matrices, the lattice $L$ is generated by the basis matrix $\otimes_{n/4}M$.

\end{description}
\end{enumerate}

\noindent $(2) \equiv (3)$: We now show that $L_C$ decomposes into direct product of lower dimensional lattices, $L_C \cong \otimes_i L_i$ if and only if the code $C$ also decomposes, $C \cong \otimes_i C_i$.

Let $L_C\cong \otimes_i L_i$. Without loss of generality, we can consider $L_C = \otimes_i L_i$. We have $C = L_C \mod 3= \otimes_i L_i \mmod 3$. 
We observe that if $L_i$ has dimension $n_i$, then $L_i\supseteq 3\Z^{n_i}$. Therefore, $C_i=L_i\mmod 3$ is a ternary code. Let $C_i := L_i \mmod 3$ for every $i$. Then $C=\otimes_i C_i$. (If $c\in C$, then $c\in L$ and hence the projection of $c$ to the subset of coordinates corresponding to $L_i$ is in $C_i$. Let $c_i\in C_i$ for every $i$. The concatenated vector $\otimes_i c_i$ is in $\otimes_i L_i \mod 3$ and hence is in $C$.) 
 %Therefore, we get $C = \otimes_i C_i$. 
 
To show the other direction of the equivalence, let $C \cong \otimes_i C_i$, where each $C_i \subseteq \mathbb{F}_3^{n_i}$ and $n = \sum_i n_i$. Therefore $L_C = C+ 3\mathbb{Z}^n  \cong \otimes_i C_i + 3\Z^n \cong \otimes_i  (C_i + 3 \mathbb{Z}^{n_i})$, since $\mathbb{Z}^n \cong \otimes_i \mathbb{Z}^{n_i}$.

%\end{enumerate}
\end{proof}

\subsection{Algorithm}\label{sec:algo}

Theorem~\ref{thm:ternary-decomposition} shows that a lattice of the form $C+ 3\mathbb{Z}^n$ is orthogonal if and only if the underlying code decomposes into direct product of ternary linear codes isomorphic to $\{ 0, 1, 2\}$ or $\{0\}$ or the four dimensional code generated by ${\cal{T}}(M) \mmod 3$, where ${\cal{T}}(M)$ is obtained from matrix $M$ by negating a subset of its columns. We now give a polynomial time algorithm which finds the decomposition of the code $C$ into the component codes, $C_i$, if there exists one. Therefore, if the lattice $L_C$ is orthogonal, the algorithm decides in polynomial time if it is orthogonal and also gives the orthogonal basis for the lattice. 

The algorithm recursively attempts to find the component codes. If it is unable to decompose the code at any stage, then it declares that $L_C$ is not orthogonal. At every step we check if $C \cong \{0,1,2\} \times C'$ or $\{0\}\times C'$ or  $C_{{\cal{T}}(M)} \times C'$ where $C_{{\cal{T}}(M)}$ is the code generated by ${\cal{T}}(M) \mmod 3$ and then recurse on $C'$.

\iffalse
\begin{algorithm}
  \caption{$\mathbf{decompose(B)}$}
  \label{algo_main}
  \textbf{Input}: $B \in \Z^{n \times n}$  ~(Basis for $L_C$) \\
  
  \begin{algorithmic}[1]
  \STATE Let $G \leftarrow B \mmod 3$
  \STATE Let $G' \leftarrow decompose01(G)$
  \IF{$G' == \emptyset$}
  	\STATE $P, G'' \leftarrow decomposeP(G)$ 
 	 \IF{$G'' == \emptyset$}
		\STATE \textbf{return } $L_C$ is not orthogonal 
	\ELSE
		\STATE \textbf{return } $P \otimes decompose(G'')$
	\ENDIF
 \ENDIF
 \STATE \textbf{return } $\{0, 1\} \otimes decompose(G')$
  \end{algorithmic}
\end{algorithm}
\knote{Maybe remove Algorithm 1. The proof of the theorem has sufficient details. If we have to write it, then we need another algorithm called decompose0(G) which seems to make the whole section unnecessarily complex.}
\fi

%\begin{theorem}\label{thm:algo3}
%Given an arbitrary basis for $L_C = C + 3\Z^n$, there exists an algorithm which runs in $O(n^8)$ time and decides if $L_C$ is orthogonal and returns the orthogonal basis if there exists one.
%\end{theorem}

\begin{proof}[Proof of Theorem \ref{thm:algo-ternary}]
Given a basis for $L_C$ as input, we first compute the generator for $C$. % compute a set of vectors whose linear closure in $\F_3^n$ 
%the gives the code the generator $G$ for the underlying code $C$. 
%gives the code $C$. 
 From Theorem~\ref{thm:ternary-decomposition}, we know that if $L_C$ is orthogonal, then $C \cong \otimes_i C_i$ where each $C_i$ is either the length-1 code $\{0, 1, 2\}$ or the length-1 code $\{0\}$ or a $4$-dimensional code generated by the rows of ${\cal{T}}(M) \mmod 3$ where ${\cal{T}}(M)$ obtained from matrix $M$ by negating a subset of its columns. 

The algorithm therefore in each step decides if $C \cong \{0,1,2\} \otimes C'$ or $C\cong \{0\}\otimes C'$ or $C \cong C_{{\cal{T}}(M)} \otimes C'$, where $C_{{\cal{T}}(M)}$ denotes the code generated by ${\cal{T}}(M) \mmod 3$.
Theorem~\ref{theorem:code-decomposition-01} shows that using Algorithm~\ref{algo1} we can check in $O(n^4)$ time, if $C \cong \{0,1,2\} \otimes C'$. The same algorithm can be modified to check in $O(n^4)$ time, if $C\cong \{0\}\otimes C'$. Theorem~\ref{theorem:code-decomposition-Cp} shows that Algorithm~\ref{algo2} can verify if $C \cong C_{{\cal{T}}(M)} \otimes C'$ in $O(n^7)$ time. If any one of the algorithms finds a decomposition, then we recurse on the lower dimensional code $C'$ to find a further decomposition. We recurse at most $n$ times. If all the algorithms 
%Algorithm~\ref{algo1} and Algorithm~\ref{algo2} 
fail to find a decomposition, then $L_C$ is not orthogonal. Therefore, it takes $O(n^8)$ time to decide if $L_C$ is orthogonal. 
\end{proof}

%We now describe the two subroutines used by the algorithm in Theorem \ref{thm:algo3}.
%Algorithm~\ref{algo_main}. 
We now describe the individual algorithms to verify if $C \cong \{0,1,2\} \otimes C'$ or $C\cong \{0\}\otimes C'$ or $C \cong C_{{\cal{T}}(M)} \otimes C'$. 

\begin{algorithm}
  \caption{$\mathbf{: decompose-length-1(G)}$:}
  \label{algo1}
  \textbf{Input}: $G=\{g_1,\ldots, g_n\} \in \F_3^{n}$ (A generator for the code $C$)\\ % ~(A set of codewords whose linear closure in $\F_3^n$ is $C$) \\

  \begin{algorithmic}[1]
  \FOR {$j \in \{ 1, \cdots, n \} $ }
  	\STATE Let $G' \leftarrow$ projection of vectors in $G$ on coordinates $[n] \setminus \{j\}$
  	\STATE For $g\in G'$, define $g^0,g^1,g^2\in \F_3^n$ as the $n$-dimensional vectors obtained by extending $g$ using $0$, $1$ and $2$ along the $j$'th coordinate respectively. %(i.e., $g^i_{j'}:=g_{j'}$ for every $j'\in[n]\setminus j$, $g^i_j:=i$).
	\IF{ $g^0,\ g^1,\ g^2 \in C$ for all $g \in G'$}
		\STATE \textbf{return } $j$
	\ENDIF
  \ENDFOR
  \STATE \textbf{return } FAIL
  \end{algorithmic}
\end{algorithm}
 
\begin{theorem}
\label{theorem:code-decomposition-01}
Let $C$ be a ternary linear code and $G=\{g_1,\ldots, g_n\}\in \F_3^{n\times n}$ be its generator. Then Algorithm \ref{algo1} decides if $C\cong \{0,1,2\}\otimes C'$ for some linear code $C'\subseteq \F_3^{n-1}$ and if so outputs the coordinate corresponding to the direct product decomposition. Moreover the algorithm runs in time $O(n^4)$.

%Let $L_C = C+3\Z^n$, for some ternary linear code $C$. If $L_C$ is orthogonal, then Algorithm~\ref{algo1} decides in $O(n^3)$ time if $C \cong \{0, 1\} \otimes C'$ for some linear code $C' \subseteq \F_3^{n-1}$. 
\end{theorem}
\begin{proof}
For $j\in [n]$, let $C_{\overline{j}}'\subseteq \F_3^{n-1}$ be the projection of $C$ on the indices $[n] \setminus \{j\}$ and for a vector $c\in C_{\overline{j}}'$, let $c^0,c^1,c^2\in \F_3^n$ be extensions of $c$ using $0,1,2$ respectively along the $j$'th coordinate. We note that $C \cong \{0,1,2\} \otimes C'$ for some ternary linear code $C'$ if and only if there exists an index  $j \in [n]$, such that  
\begin{equation*} \label{check1} 
C = \left\{ c^0, c^1, c^2 \mid  c \in C_{\overline{j}}' \right\}. 
\end{equation*}
From the definition of $C_{\overline{j}}'$, it follows that $C\subseteq \{ c^0, c^1, c^2 \mid c \in C_{\overline{j}'} \}$ up to a permutation of coordinates. So, the algorithm just needs to verify if the other side of the containment holds for some $j$.

%Let $G$ be the generator for the code $C$ %a set of codewords whose linear closure in $\F_3^n$ is the code $C$ (i.e., $G$ is a superset of the generators of the code $C$). 
Let $G'$ be the set of vectors of $G$ projected on the coordinates $[n] \setminus \{j\}$. Algorithm~\ref{algo1} verifies if $g^0$, $g^1$ and $g^2$ are codewords in $C$, for every vector $g \in G'$. We now show that this is sufficient.
Since $C$ is a code, if $g^0$, $g^1$, $g^2$ $\in C$ for every $g\in G'$, then all linear combinations of these vectors are also in $C$.  Therefore, $\{ c^0, c^1 , c^2 \mid c \in C_{\overline{j}}' \} \subseteq C$.

It takes $O(n^2)$ %\gnote{Should this be $O(n^2)$ since we have to do a matrix-vector product?}\knote{Doesn't this check need the parity matrix?	}
time to compute a parity check matrix from the generator $G$ and $O(n^2)$ time to verify if an input vector is a codeword using the parity check matrix.  
For every possible choice of the index $j$, Algorithm~\ref{algo1} checks if each of the $3n$ vectors of the form $g^0,g^1,g^2$ are in $C$. Therefore, Algorithm~\ref{algo1} takes $O(n^4)$ time to decide if $C \cong \{0, 1, 2\} \otimes C'$. 
\end{proof}

\begin{algorithm}
  \caption{ $\mathbf{: decompose-length-4(G)}$:}
  \label{algo2}
  \textbf{Input}: $G=\{g_1,\ldots, g_n\} \in \F_3^{n}$ (A generator for the code $C$)\\
	
  \begin{algorithmic}[1]
  %\STATE Let $M \leftarrow \begin{bmatrix} TODO \end{bmatrix}$
  \FOR{ $j_1, j_2, j_3, j_4 \in \{1, 2, \cdots, n \} $}
	\STATE Let $G' \leftarrow$ projection of vectors in $G$ on coordinates $[n] \setminus \{ j_1, j_2, j_3, j_4 \}$
	\STATE Let $G'' \leftarrow$ projection of vectors in $G$ on coordinates $\{ j_1, j_2, j_3, j_4 \}$
  	\FOR{ $S \subseteq [4]$}
		\STATE Let ${\cal{T}}(M) \leftarrow M$ with columns in $S$ negated
		\IF{ $C_{{\cal{T}}(M)} \equiv$ Code generated by $G''$}
			\STATE For $g \in G'$ define $g^{p_1}, g^{p_2}, g^{p_3}, g^{p_4} \in \F_3^n$ be $n$-dimensional vectors obtained by extending $g$ using the rows of ${\cal{T}}(M)$ along the $j_1, j_2, j_3, j_4$ coordinates. 
			\IF{ $g^{p_1}, g^{p_2}, g^{p_3}, g^{p_4} \in C$ for all $g \in G'$}
				\STATE \textbf{ return } $j_1, j_2, j_3, j_4$ and ${\cal{T}}(M)$
			\ENDIF
		\ENDIF
	\ENDFOR
  \ENDFOR
  \STATE \textbf{ return } FAIL

  \end{algorithmic}
\end{algorithm}

\begin{theorem}
\label{theorem:code-decomposition-Cp}
Let $C$ be a ternary linear code and $G=\{g_1,\ldots, g_n\}\in \F_3^{n\times n}$ be its generator. For a matrix ${\cal{T}}(M)$ obtained by negating a subset of columns of $M$, let $C_{{\cal{T}}(M)}$ be the length-$4$ code whose generators are the rows of ${\cal{T}}(M)$. 
Then Algorithm \ref{algo2} decides if $C\cong C_{{\cal{T}}(M)} \otimes C'$ for some linear codes $C'\subseteq \F_3^{n-4}$ and $C_{{\cal{T}}(M)} \subseteq \F_3^4$ and if so outputs the coordinates corresponding to the direct product decomposition as well as the matrix ${\cal{T}}(M)$. 
Moreover the algorithm runs in time $O(n^7)$.
%Let $L_C = C+3\Z^n$, for some ternary linear code $C$ and $C_P$ be the code generated by the rows of the matrix  $P$, where $P$ is obtained by negating a subset of columns of $M$. If $L_C$ is orthogonal, then Algorithm~\ref{algo2} decides in $O(n^8)$ time if $C \cong C_P \otimes C'$ for some ternary linear code $C'$. 
\end{theorem}
\begin{proof}
For $1 \leq j_1 < j_2 < j_3 < j_4 \leq n$, let $C''_{j_1, j_2, j_3, j_4}$ be the projection of $C$ on the indices $\{j_1, j_2, j_3, j_4\}$. We first verify if $C''_{ j_1, j_2, j_3, j_4}$ is the code generated by the rows of ${\cal{T}}(M)$ (denoted as $C_{{\cal{T}}(M)}$) for some ${\cal{T}}(M)$ which is obtained by negating a subset of columns of $M$. We would like to check if every $c \in C''_{ j_1, j_2, j_3, j_4}$ is in $C_{{\cal{T}}(M)}$ and vice versa. For this purpose, it is sufficient to check if the generator vectors of $C''_{j_1, j_2, j_3, j_4}$ are codewords in $ C_{{\cal{T}}(M)}$ and each row of ${\cal{T}}(M)$ is a codeword in  $C''_{j_1, j_2, j_3, j_4}$. We know that the generators of $C''_{ j_1, j_2, j_3, j_4}$ are contained in $G''$ where $G''$ is the set of vectors in $G$ projected on the indices $\{j_1, j_2, j_3, j_4\}$. 

Once we fix ${\cal{T}}(M)$ such that $C''_{j_1, j_2, j_3, j_4} = C_{{\cal{T}}(M)}$, to see if $C \cong C_{{\cal{T}}(M)} \otimes C'$ for some ternary linear code $C' \subseteq \F_3^{n-4}$. Define $C'_{\bar{j_1}, \bar{j_2}, \bar{j_3}, \bar{j_4} }$ to be the projection of C on the indices $[n] \setminus \{j_1, j_2, j_3, j_4\}$. For a vector $c \in C'_{\bar{j_1}, \bar{j_2}, \bar{j_3}, \bar{j_4}}$, let $c^{p} \in \F_3^n$ be the extensions of $c$ using a codeword $p \in C_{{\cal{T}}(M)}$ along the $j_1, j_2, j_3, j_4$ coordinates. 
We note that $C \cong C_{{\cal{T}}(M)} \otimes C'$ for some ternary linear code $C'$ if and only if there exist indices  $j_1, j_2, j_3, j_4 \in [n]$, such that  
\begin{equation} \label{check2} 
C = \left\{ c^{p} \mid c \in C_{\bar{j_1}, \bar{j_2}, \bar{j_3}, \bar{j_4}}' , p \in C_{{\cal{T}}(M)} \right\}. 
\end{equation}
From the definition of $C_{\bar{j_1}, \bar{j_2}, \bar{j_3}, \bar{j_4}}' $ and $C''_{j_1, j_2, j_3, j_4} ~(=C_{{\cal{T}}(M)})$, it follows that $C\subseteq \{ c^{p} \mid c \in C_{\bar{j_1}, \bar{j_2}, \bar{j_3}, \bar{j_4}}', p \in C_{{\cal{T}}(M)} \}$. So, the algorithm just needs to verify if the other side of the containment holds for some indices $j_1, j_2, j_3, j_4 \in [n]$.

%Let $G$ be the generator for the code $C$ %a set of codewords whose linear closure in $\F_3^n$ is the code $C$ (i.e., $G$ is a superset of the generators of the code $C$). 
Let $G'$ be the set of vectors of $G$ projected on the coordinates $[n] \setminus \{j_1, j_2, j_3, j_4\}$. Algorithm~\ref{algo2} verifies if $g^{p_0}$, $g^{p_1}$, $g^{p_3}$ and $g^{p_4}$ are codewords in $C$, for every vector $g \in G'$. We now show that this is sufficient.
Since $C$ is a code, if $g^{p_0}, g^{p_1}, g^{p_3}, g^{p_4} \in C$ for every $g\in G'$ and $p_i \in {\cal{T}}(M)$, then all linear combinations of these vectors are also in $C$.  Therefore, $\{ c^{p} \mid c \in C_{\bar{j_1}, \bar{j_2}, \bar{j_3}, \bar{j_4}}', p \in C_{{\cal{T}}(M)} \} \subseteq C$. 

There are $2^4 4^4$ possible choices of ${\cal{T}}(M)$ including permutations.  For each matrix ${\cal{T}}(M)$, it takes $O(n)$ time to verify if $C_{{\cal{T}}(M)} = C_{j_1, j_2, j_3, j_4}''$.  Time to verify if an input vector is a codeword using the parity check matrix is $O(n^2)$.  We perform this check for $4n$ vectors of the form $\{ g^{p_0}, g^{p_1}, g^{p_3}, g^{p_4} \mid g  \in G' \}$. So,  for a given ${\cal{T}}(M)$ such that $C_{{\cal{T}}(M)} = C_{j_1, j_2, j_3, j_4}''$, It takes $O(n^3)$ time to verify $C \cong C_{{\cal{T}}(M)} \otimes C'$.

For every possible choice of indices, $\{j_1, j_2, j_3, j_4\}$, Algorithm~\ref{algo2} takes $O(n^3)$ time to verify if $C \cong C_{{\cal{T}}(M)} \otimes C_{\bar{j_1}, \bar{j_2}, \bar{j_3}, \bar{j_4}}'$. Since there are at most $n \choose 4$ possible choices of indices, it takes $O(n^7)$ time in total to decide if $C \cong C_{{\cal{T}}(M)} \otimes C'$.

\end{proof}

\section{Acknowledgments}
We thank Daniel Dadush for helpful suggestions and pointers.
\bibliographystyle{abbrv}

\bibliographystyle{abbrv}
\bibliography{PropertyTest-Bib}

%\appendix
%\input{appendix}

%\input{appendix}

\end{document}